\documentclass[fleqn,12pt,twoside]{article}

\usepackage[headings]{espcrc1}
\readRCS $Id: espcrc1.tex,v 1.2 2004/02/24 11:22:11 spepping Exp $
\usepackage{graphicx}
\usepackage[figuresright]{rotating}

\newtheorem{theorem}{Theorem}

\newtheorem{claim}{Claim}

\newtheorem{corollary}{Corollary}

\newtheorem{lemma}{Lemma}

\newtheorem{remark}{Remark}

\newenvironment{proof}[1][Proof.]{\begin{trivlist}
\item[\hskip \labelsep {\bfseries #1}]}{\end{trivlist}}

\newcommand{\AmS}{{\protect\the\textfont2
  A\kern-.1667em\lower.5ex\hbox{M}\kern-.125emS}}

\usepackage{amsmath}
\usepackage{amsfonts}
\title{On upper bounds for parameters related to construction of special maximum matchings
}

\author{Artur Khojabaghyan\address[MCSD]{Department of Informatics and Applied Mathematics,\\
Yerevan State University, Yerevan, 0025, Armenia}
\thanks{email: arturkhojabaghyan@gmail.com}
                                and
        Vahan V. Mkrtchyan\addressmark[MCSD]
        \address{Institute for Informatics and Automation Problems,\\
National Academy of Sciences of Republic of Armenia, 0014, Armenia}
\thanks{email: vahanmkrtchyan2002@\{ysu.am, ipia.sci.am,
yahoo.com\}}}

\runtitle{On upper bounds for parameters related to construction of special maximum matchings}
\runauthor{Artur Khojabaghyan, Vahan V. Mkrtchyan}

\begin{document}

% typeset front matter
\maketitle

\begin{abstract}
For a graph $G$ let $L(G)$ and $l(G)$ denote the size of the largest and
smallest maximum matching of a graph obtained from $G$ by removing a
maximum matching of $G$. We show that $L(G)\leq 2l(G),$ and $L(G)\leq \frac{3}{2}%
l(G)$ provided that $G$ contains a perfect matching. We also characterize
the class of graphs for which $L(G)=2l(G)$. Our characterization implies the
existence of a polynomial algorithm for testing the property $L(G)=2l(G)$.
Finally we show that it is $NP$-complete to test whether a graph $G$
containing a perfect matching satisfies $L(G)=\frac{3}{2}l(G)$.
\end{abstract}

\section{Introduction}

In the paper graphs are assumed to be finite, undirected, without loops or
multiple edges. Let $V(G)$ and $E(G)$ denote the sets of vertices and edges
of a graph $G$, respectively. If $v\in V(G)$ and $e\in E(G)$, then $e$ is said to cover $v$ if $e$ is incident to $v$. For $V'\subseteq V(G)$ and $E'\subseteq E(G)$ let $G\backslash V'$ and $G\backslash E'$ denote the graphs obtained from $G$ by removing $V'$ and $E'$, respectively. Moreover, let $V(E')$ denote the set of vertices of $G$ that are covered by an edge from $E'$. A subgraph $H$ of $G$ is said to be spanning for $G$, if $V(E(H))=V(G)$.

The length of a path (cycle) is the number of its edges. A $k$-path ($k$-cycle) is a
path (cycle) of length $k$. A $3$-cycle is called a triangle.

A set $V'\subseteq V(G)$ ($E'\subseteq E(G)$) is said to be independent, if $V'$ ($E'$) contains no adjacent vertices (edges). An independent set of edges is called matching. A matching of $G$ is called perfect, if it covers all vertices of $G$. Let $\nu (G)$ denote the
cardinality of a largest matching of $G$. A matching of $G$ is maximum, if it contains $\nu (G)$ edges.

For a positive integer $k$ and a matching $M$ of $G$, a $(2k-1)$-path $P$ is called $M$-augmenting, if the $2^{nd}$, $4^{th}$, $6^{th}$,..., $(2k-2)^{th}$ edges of $P$ belong to $M$, while the endvertices of $P$ are not covered by an edge of $M$. The following theorem of Berge gives a sufficient and necessary condition for a matching to be maximum:

\begin{theorem} (Berge \cite{Harary}) A matching $M$ of $G$ is maximum, if $G$ contains no $M$-augmenting path.
\end{theorem}   

For two matchings $M$ and $M'$ of $G$ consider the subgraph $H$ of $G$, where $V(H)=V(M\triangle M')$ and $E(H)=M\triangle M'$. The connected components of $H$ are called $M\triangle M'$-alternating components. Note that $M\triangle M'$ alternating components are always paths or cycles of even length. For a graph $G$ define:

\begin{center}
$L(G)\equiv \max \{\nu (G\backslash F):F$ is a maximum matching of $G\},$

$l(G)\equiv \min \{\nu (G\backslash F):F$ is a maximum matching of $G\}.$
\end{center}

It is known that $L(G)$ and $l(G)$ are $NP$-hard calculable even for
connected bipartite graphs $G$ with maximum degree three \cite{Complexity},
though there are polynomial algorithms which construct a maximum matching $F$
of a tree $G$ such that $\nu (G\backslash F)=L(G)$ and $\nu (G\backslash
F)=l(G)$ (to be presented in \cite{Algorithm}).

In the same paper \cite{Algorithm} it is shown that $L(G)\leq 2l(G).$ In the
present paper we re-prove this equality, and also show that $L(G)\leq \frac{3%
}{2}l(G)$ provided that $G$ contains a perfect matching. 

A naturally arising question is the characterization of graphs $G$ with $L(G)= 2l(G)$ and the graphs $G$ with a perfect matching that satisfy $L(G)= \frac{3}{2}l(G)$. In this paper we solve these problems by giving a characterization of graphs $G$ with $L(G)= 2l(G)$ that implies the existence of a polynomial algorithm for testing this property, and by showing that it is $NP$-complete to test whether a bridgeless cubic graph $G$ satisfies $L(G)=\frac{3}{2}l(G)$. Recall that by Petersen theorem any bridgeless cubic graph contains a perfect matching (see, for example, theorem 3.4.1 of \cite{Lov}).

Terms and concepts that we do not define can be found in \cite%
{Diestel,Harary,Lov,West}.

\section{Some auxiliarly results}

We will need the following:

\begin{theorem}\label{Ratios} Let $G$ be a graph. Then:

\begin{enumerate}
\item[(a)]for any two maximum matchings $F,F'$ of $G$, we have $\nu (G\backslash F')\leq 2\nu (G\backslash F)$;

\item[(b)] $L(G)\leq 2l(G)$;

\item[(c)] If $L(G)=2l(G)$, $F_{L},F_{l}$
are two maximum matchings of the graph $G$ with $\nu (G\backslash
F_{L})=L(G), $ $\nu (G\backslash F_{l})=l(G)$, and $H_{L}$ is
\textbf{any} maximum matching of the graph $G\backslash F_{L}$, then:
\begin{enumerate}
\item[(c1)] $F_{l}\backslash F_{L}\subset H_{L};$

\item[(c2)] $H_{L}\backslash F_{l}$ is a maximum matching of $G\backslash
F_{l}$;

\item[(c3)] $F_{L}\backslash F_{l}$ is a maximum matching of $G\backslash
F_{l}$;
\end{enumerate}

\item[(d)] if $G$ contains a perfect matching, then $L(G)\leq \frac32l(G)$.
\end{enumerate}

\end{theorem}

\begin{proof}(a)Let $H'$ be any maximum matching in the graph $G\backslash F'$. Then:%
\begin{eqnarray*}
\nu (G\backslash F') =|H'|=|H'\cap F|+|H'\backslash F|\leq |F\backslash F'|+\nu(G\backslash F)=
|F'\backslash F|+\nu (G\backslash F)\leq 2\nu (G\backslash F).
\end{eqnarray*}

(b) follows from (a).

(c) Consider the proof of (a) and take $F'=F_{L}$, $H'=H_L$ and $F=F_l$. Since $L(G)=2l(G)$,
we must have equalities throughout, thus properties (c1)-(c3) should be true.

(d) Let $F_{L},F_{l}$ be two perfect matchings of the graph $G$ with $\nu
(G\backslash F_{L})=L(G),$ $\nu (G\backslash F_{l})=l(G)$, and assume $H_{L}$
to be a maximum matching of the graph $G\backslash F_{L}$. Define:
\begin{eqnarray*}
X =\{e=(u,v)\in F_{L}:u\text{ and }v\text{ are incident to an edge from }H_{L}\cap F_{l}\}, \\
x =|X|, k=|H_{L}\cap F_{l}|;
\end{eqnarray*}

Clearly, $(H_{L}\backslash F_{l})\cup X$ is a matching of the graph $G\backslash F_{l}$, therefore, taking into account that $(H_{L}\backslash F_{l})\cap X=\emptyset,$ we deduce
\begin{eqnarray*}
l(G) =\nu (G\backslash F_{l})\geq |H_{L}\backslash
F_{l}| +|X| =|H_{L}|
-|H_{L}\cap F_{l}| +|X| =L(G)-k+x.
\end{eqnarray*}

Since $F_{L}$ is a perfect matching, it covers the set $V(H_{L}\cap
F_{l})\backslash V(X)$, which contains%
\begin{equation*}
|V(H_{L}\cap F_{l})\backslash V(X)| =2|
(H_{L}\cap F_{l})| -2|X| =2k-2x
\end{equation*}%
vertices. Define the set $E_{F_{L}}$ as follows:
\begin{equation*}
E_{F_{L}}=\{e\in F_{L}:e\text{ covers a vertex from }V(H_{L}\cap
F_{l})\backslash V(X)\}.
\end{equation*}%

Clearly, $E_{F_{L}}$is a matching of $G\backslash F_{l}$, too, and therefore
\begin{equation*}
l(G)=\nu (G\backslash F_{l})\geq |E_{F_{L}}|=2k-2x.
\end{equation*}%

Let us show that
\begin{equation*}
\max \{L(G)-k+x,2k-2x\}\geq \frac{2L(G)}{3}.
\end{equation*}%
Note that \\
if $x\geq k-\frac{L(G)}{3}$ then $L(G)-k+x\geq L(G)-k+k-\frac{L(G)}{3}=\frac{2L(G)}{3}$; \\
if $x\leq k-\frac{L(G)}{3}$ then $2k-2x\geq \frac{2L(G)}{3}$, \\
thus in both cases we have $l(G)\geq \frac{2L(G)}{3}$ or
\begin{equation*}
\frac{L(G)}{l(G)}\leq \frac{3}{2}.
\end{equation*}

The proof of the theorem \ref{Ratios} is completed.
$\square$
\end{proof}

\begin{lemma}
\label{1-1Case} (Lemma 2.20, 2.41 of \cite{Algorithm}) Let $G$ be a graph,
and assume that $u$ and $v$ are vertices of degree one sharing a neighbour $w\in V(G)$. Then:
\begin{equation*}
L(G)=L(G\backslash \{u,v,w\})+1,l(G)=l(G\backslash \{u,v,w\})+1.
\end{equation*}
\end{lemma}

\begin{proof}The proofs of these two equalities are similar, thus we will stop only on the proof of the first one. Our proof is based on the ideas of \cite{Algorithm}.

First of all, observe that $\nu(G)=\nu(G\backslash \{u,v,w\})+1$. Let us show that $L(G)\geq L(G\backslash \{u,v,w\})+1$. Take any maximum matching $F$ of $G\backslash \{u,v,w\}$ with $\nu((G\backslash \{u,v,w\})\backslash F)=L(G\backslash \{u,v,w\})$, and let $H$ be any maximum matching of $(G\backslash \{u,v,w\})\backslash F$. Define:
\begin{equation*}
F'=F\cup \{(u,w)\}, H'=H\cup \{(v,w)\}.
\end{equation*}
Observe that $F'$ is a maximum matching of $G$, and $H'$ is a matching of $G\backslash F'$. Thus,
\begin{equation*}
L(G)\geq \nu(G\backslash F')\geq |H'|=1+|H|=1+\nu((G\backslash \{u,v,w\})\backslash F)=1+L(G\backslash \{u,v,w\}).
\end{equation*}

To complete the proof of the first equality, it suffices to show that $L(G)\leq L(G\backslash \{u,v,w\})+1$. First of all, let us show that there is a maximum matching $F'$ of $G$ with $\nu(G\backslash F')=L(G)$, such that $F'$ contains one of the edges $(u,w)$ and $(v,w)$.

Take any maximum matching $F'$ of $G$ with $\nu(G\backslash F')=L(G)$, and assume that $F'\cap \{(u,w), (v,w)\}=\emptyset$. Morever, let $H'$ be a maximum matching of $G\backslash F'$. Since $F'$ is a maximum matching of $G$, $F'$ must contain an edge $(w,z)$, where $z\neq u,v$. Define:
\begin{equation*}
F''=\left\{
\begin{array}{ll}
(F'\backslash \{(w,z)\})\cup \{(w,u)\}, & \text{if } (w,u)\notin H;\\
(F'\backslash \{(w,z)\})\cup \{(w,v)\}, & \text{if } (w,u)\in H.
\end{array}
\right.
\end{equation*}        
Observe that $F''$ is a maximum matching of $G$, and $H'$ is a matching of $G\backslash F''$. Thus,
\begin{equation*}
\nu(G\backslash F'')\geq |H'|=\nu(G\backslash F')=L(G).
\end{equation*}
The last inequality implies that $\nu(G\backslash F'')=L(G)$. Moreover, $F''\cap \{(u,w), (v,w)\}\neq \emptyset$.

Thus, initially we can assume that $F'$ is a maximum matching of $G$ with $\nu(G\backslash F')=L(G)$, such that $F'$ contains one of the edges $(u,w)$ and $(v,w)$. Without loss of generality, we can also assume that this edge is $(u,w)$. Now, we claim that there is a maximum matching $H'$ of $G\backslash F'$ that contains the edge $(v,w)$.

Take any maximum matching $H'$ of $G\backslash F'$ and suppose that $(v,w)\notin H'$. Since $H'$ is a maximum matching of $G\backslash F'$, there must exist an edge $(w,z)\in H'$, where $z\neq u,v$. Define:
\begin{equation*}
H''=(H'\backslash \{(w,z)\})\cup \{(w,v)\}.
\end{equation*}
Observe that $H''$ is a maximum matching of $G\backslash F'$, since $|H''|=|H'|=\nu(G\backslash F')=L(G)$. Moreover, it contains the edge $(w,v)$. Thus, initially we can assume that $H'$ is a maximum matching of $G\backslash F'$ that contains the edge $(v,w)$.

We are ready to show that $L(G)\leq L(G\backslash \{u,v,w\})+1$. Since $\nu(G)=\nu(G\backslash \{u,v,w\})+1$, we have that $F'\backslash \{(u,w)\}$ is a maximum matching of $G\backslash \{u,v,w\}$. Taking into account that $H'\backslash \{(v,w)\}$ is a matching of $(G\backslash \{u,v,w\})\backslash (F'\backslash \{(u,w)\})$, we deduce:
\begin{eqnarray*}
L(G)=\nu(G\backslash F')=|H'|=1+|H'\backslash \{(v,w)\}|\leq 1+\nu((G\backslash \{u,v,w\})\backslash (F'\backslash \{(u,w)\}))\leq \\
 1+L(G\backslash \{u,v,w\}).
\end{eqnarray*}
$\square$
\end{proof}

\begin{corollary}\label{L=2l1-1case} Let $G$ be a graph with $L(G)=2l(G)$. Then there are no vertices $u,v$ of degree one, that are adjacent to the same vertex $w$.
\end{corollary}
\begin{proof}
Suppose not. Then lemma \ref{1-1Case} and (b) of theorem \ref{Ratios} imply
\begin{equation*}
L(G)=1+L(G-\{u,v,w\})\leq 1+2l(G-\{u,v,w\})=1+2(l(G)-1)<2l(G)
\end{equation*}
a contradiction.
$\square$
\end{proof}

\section{Characterization of graphs $G$ satisfying $L(G)=2l(G)$}

Let $T$ be the set of all triangles of $G$ that contain at least two vertices of degree two.
Note that any vertex of degree two lies in at most one triangle from $T$. From each triangle $t\in T$
choose a vertex $v_t$ of degree two, and define $V_1(G)$ as follows:
\begin{equation*}
V_1(G)=\{v:d_G(v)=1\}\cup \{v_t:t\in T\}
\end{equation*}

\begin{theorem}
\label{Characterization} Let $G$ be a connected graph with $|V(G)|\geq3$. Then $L(G)=2l(G)$ if and only if

\begin{description}
\item[(1)] $G\backslash V_1(G)$ is a bipartite graph with a bipartition $(X,Y)$;

\item[(2)] $|V_1(G)|=|Y|$ and any $y\in Y$ has exactly one neighbour in $V_1(G)$;

\item[(3)] the graph $G\backslash V_{1}(G)$ contains $|X|$ vertex disjoint $2$-paths.
\end{description}
\end{theorem}

\begin{proof}
Sufficiency. Let $G$ be a connected graph with $|V(G)|\geq3$ satisfying the conditions (1)-(3).
Let us show that $L(G)=2l(G)$.

For each vertex $v$ with $d(v)=1$ take the edge incident to it and define $F_1$
as the union of all these edges. For each vertex $v_t\in V_{1}(G)$ take the edge that
connects $v_t$ to a vertex of degree two, and define $F_2$ as the union of all those edges.
Set:
\begin{equation*}
F=F_1\cup F_2.
\end{equation*}

Note that $F$ is a matching with $|F|=|V_{1}(G)|=|Y|$. Moreover, since $G$
is bipartite and $|V_{1}(G)|=|Y|$, the definitions of $F_1$ and $F_2$ imply
that there is no $F$-augmenting path in $G$. Thus, by Berge theorem, $F$ is a maximum
matching of $G$, and
\begin{equation*}
\nu(G)=|F|=|V_{1}(G)|=|Y|.
\end{equation*}Observe that the graph $G\backslash F$ is a bipartite graph with $\nu(G\backslash F)\leq |X|$,
thus
\begin{equation*}
l(G)\leq \nu(G\backslash F)\leq |X|.
\end{equation*}

Now, consider the $|X|$ vertex disjoint $2$-paths of the graph $G\backslash V_{1}(G)$ guaranteed by (3).
(2) implies that these $2$-paths together with the $|F|=|V_{1}(G)|=|Y|$ edges of $F$ form $|X|$ vertex
disjoint $4$-paths of the graph $G$.

Consider matchings $M_{1}$ and $M_{2}$ of $G$ obtained from these $4$-paths by adding the first and the third, the
second and the fourth edges of these $4$-paths to $M_{1}$ and $M_{2}$, respectively. Define:
\begin{equation*}
F'=(F\backslash M_{2})\cup (M_{1}\backslash F).
\end{equation*}

Note that $F'$ is a matching of $G$ and $|F'|=|F|$, thus $F'$ is a maximum
matching of $G$. Since $F'\cap M_2=\emptyset$, we have%
\begin{equation*}
L(G)\geq \nu (G\backslash F')\geq |M_{2}|=2|X|\geq 2l(G).
\end{equation*}

(b) of theorem \ref{Ratios} implies that $L(G)=2l(G)$.

Necessity. Now, assume that $G$ is a connected graph with $|V(G)|\geq 3$ and $L(G)=2l(G)$. By
proving a series of claims, we show that $G\backslash V_1(G)$ satisfies the conditions (1)-(3) of the theorem.

\begin{claim}
\label{SpanningSubgraph} For any maximum matchings $F_{L},F_{l}$ of the
graph $G$ with $\nu (G\backslash F_{L})=L(G),$ $\nu (G\backslash F_{l})=l(G)$,
$F_{L}\cup F_{l}$ induces a spanning subgraph, that is $V(F_{L})\cup
V(F_{l})=V(G)$.
\end{claim}

\begin{proof}
Suppose that there is a vertex $v\in V(G)$ that is covered neither by $F_{L}$
nor by $F_{l}$. Since $F_{L}$ and $F_{l}$ are maximum matchings of $G$, for
each edge $e=(u,v)$ the vertex $u$ is incident to an edge from $F_{L}$ and
to an edge from $F_{l}$.

Case 1: there is an edge $e=(u,v)$ such that $u$ is incident to an edge from
$F_{L}\cap F_{l}$.

Note that $\{e\}\cup (F_{L}\backslash F_{l})$ is a matching of $G\backslash
F_{l}$ which contradicts (c3) of the theorem \ref{Ratios}.

Case 2: for each edge $e=(u,v)$ $u$ is incident to an edge $f_{L}\in
F_{L}\backslash F_{l}$ and to an edge $f_{l}\in F_{l}\backslash F_{L}$.

Let $H_{L}$ be any maximum matching of $G\backslash F_{L}$. Due to (c1) of
theorem \ref{Ratios} $f_{l}\in H_{L}$. Define:
\begin{equation*}
H'_{L}=(H_{L}\backslash \{f_{l}\})\cup \{e\}.
\end{equation*}%

Note that $H'_{L}$ is a maximum matching of $G\backslash F_{L}$
such that $F_{l}\backslash F_{L}$ is not a subset of $H'_{L}$ contradicting (c1) of
theorem \ref{Ratios}.
$\square$
\end{proof}

\begin{claim}
\label{AltComp2paths} For any maximum matchings $F_{L},F_{l}$ of the graph $G
$ with $\nu (G\backslash F_{L})=L(G),$ $\nu (G\backslash F_{l})=l(G)$, the
alternating components $F_{L}\triangle F_{l}$ are $2$-paths.
\end{claim}

\begin{proof}
It suffices to show that there is no edge $f_{L}\in F_{L}$ that is adjacent
to two edges from $F_{l}$. Suppose that some edge $f_{L}\in F_{L}$ is
adjacent to edges $f_{l}^{\prime }$ and $f_{l}^{\prime \prime }$ from $F_{l}$.
Let $H_{L}$ be any maximum matching of $G\backslash F_{L}$. Due to (c1) of
theorem \ref{Ratios} $f'_{l},f''_{l}\in H_{L}$. This implies that $\{f_{L}\}\cup (H_{L}\backslash F_{l})$ is a
matching of $G\backslash F_{l}$ which contradicts (c2) of theorem \ref{Ratios}.
$\square$
\end{proof}

\begin{claim}
\label{DegreeRequirements}For any maximum matchings $F_{L},F_{l}$ of the
graph $G$ with $\nu (G\backslash F_{L})=L(G),$ $\nu (G\backslash F_{l})=l(G)$

\begin{enumerate}
\item[(a)] if $u\in V(F_{l})\backslash V(F_{L})$ then $d(u)=1$ or $d(u)=2$. Moreover, in the latter case, if $v$ and $w$ denote the two neighbours of $u$, where $(u,w)\in F_l$, then $d(w)=2$ and $(v,w)\in F_L$.

\item[(b)] if $u\in V(F_{L})\backslash V(F_{l})$ then $d(u)\geq 2.$

\end{enumerate}
\end{claim}

\begin{proof}
(a) Assume that $u$ is covered by an edge $e_{l}\in F_{l}$ and $u\notin
V(F_{L})$. Suppose that $d(u)\geq 2$, and there is an edge $e=(u,v)$ such
that $e\notin F_{l}$. Taking into account the claim \ref{SpanningSubgraph},
we need only to consider the following four cases:

Case 1: $v\in V(F_{l})\backslash V(F_{L})$.

This is impossible, since $F_{L}$ is a maximum matching.

Case 2: $v$ is covered by an edge $f\in F_{L}\cap F_{l}$;

Let $H_{L}$ be any maximum matching of $G\backslash F_{L}$. Due to (c1) of
theorem \ref{Ratios} $e_{l}\in H_{L}$, thus $e\notin H_{L}$.

Define:
\begin{equation*}
F_{L}^{\prime }=(F_{L}\backslash \{f\})\cup \{e\}.
\end{equation*}%
Note that $F_{L}^{\prime }$ is a maximum matching, and $H_{L}$ is a matching
of $G\backslash F_{L}^{\prime }$. Moreover,
\begin{equation*}
\nu (G\backslash F_{L}^{\prime })\geq \left\vert H_{L}\right\vert =\nu
(G\backslash F_{L})=L(G),
\end{equation*}
thus $H_{L}$ is a maximum matching of $G\backslash F_{L}^{\prime }$ and $\nu
(G\backslash F_{L}^{\prime })=L(G)$. This is a contradiction because $%
F_{L}^{\prime }\triangle F_{l}$ contains a component which is not a $2$%
-path contradicting claim \ref{AltComp2paths}.

Case 3: $v$ is incident to an edge $f_{L}\in F_{L},$ $f_{l}\in F_{l}$ and $%
f_{L}\neq $ $f_{l}$.

Let $H_{L}$ be any maximum matching of $G\backslash F_{L}$. Due to (c1) of
theorem \ref{Ratios}, $e_{l},f_{l}\in H_{L}$. Define:
\begin{equation*}
F_{L}^{\prime }=(F_{L}\backslash \{f_{L}\})\cup \{e\}.
\end{equation*}%
Note that $F_{L}^{\prime }$ is a maximum matching, and $H_{L}$ is a matching
of $G\backslash F_{L}^{\prime }$. Moreover,
\begin{equation*}
\nu (G\backslash F_{L}^{\prime })\geq \left\vert H_{L}\right\vert =\nu
(G\backslash F_{L})=L(G),
\end{equation*}%
thus $H_{L}$ is a maximum matching of $G\backslash F_{L}^{\prime }$ and $\nu
(G\backslash F_{L}^{\prime })=L(G)$. This is a contradiction because $%
F_{L}^{\prime }\triangle F_{l}$ contains a component which is not a $2$%
-path contradicting claim \ref{AltComp2paths}.

Case 4: $v$ is covered by an edge $e_{L}\in F_{L}$ and $v\notin V(F_{l}).$

Note that if $e_{L}$ is not adjacent to $e_{l}$ then claim \ref{AltComp2paths} implies that the edges $e,e_{L}$ and
the edge $\tilde{e}\in F_{l}\backslash F_{L}$ that is adjacent to $e_{L}$
would form an augmenting $3$-path with respect to $F_{L}$, which would
contradict the maximality of $F_{L}$.

Thus it remains to consider the case when $e_{L}$ is adjacent to $e_{l}$ and
$d(u)=2$. Let $w$ be the vertex adjacent to both $e_{l}$ and $e_{L}$.
Let us show that $d(w)=2$. Let $H_{L}$ be any maximum
matching of $G\backslash F_{L}$. Due to (c1) of theorem \ref{Ratios}, $e_{l}\in H_{L}$. Define:
\begin{equation*}
F'_{L}=(F_{L}\backslash \{e_{L}\})\cup \{e\}.
\end{equation*}%
Note that $F'_{L}$ is a maximum matching, and $H_{L}$ is a matching
of $G\backslash F'_{L}$. Moreover,
\begin{equation*}
\nu (G\backslash F'_{L})\geq |H_{L}| =\nu (G\backslash F_{L})=L(G),
\end{equation*}%
thus $H_{L}$ is a maximum matching of $G\backslash F'_{L}$ and $\nu
(G\backslash F'_{L})=L(G)$. If $d(w)\geq 3$ there is a vertex $%
w'\neq u,v$ such that $(w,w')\in E(G)$ and $w'$
satisfies one of the conditions of cases 1,2 and 3 with respect to $%
F'_{L}$ and $F_{l}$. A contradiction. Thus $d(w)=2$.

Clearly, $(v,w)=e_L\in F_L$.

(b) This follows from (a) of claim \ref{DegreeRequirements} and corollary \ref{L=2l1-1case}.
$\square$
\end{proof}

\begin{claim}
\label{IntersectionEdges} Let $F_{L},F_{l}$ be any maximum matchings of the
graph $G$ with $\nu (G\backslash F_{L})=L(G),$ $\nu (G\backslash F_{l})=l(G)$.
Then for any maximum matching $H_{L}$ of the graph $%
G\backslash F_{L}$ there is no edge of $F_{L}\cap F_{l}$ which is adjacent
to two edges from $H_{L}$.
\end{claim}

\begin{proof}
Due to (c3) of theorem \ref{Ratios} any edge from $H_{L}$ that is incident to a vertex covered by an edge of $F_{L}\cap F_{l}$ is also incident to a vertex from $V(F_{L})\backslash
V(F_{l})$. If there were an edge $e\in F_{L}\cap F_{l}$ which is adjacent to
two edges $h_{L},h_{L}^{\prime }\in H_{L}$, then the edges $h_{L},e$ and $h_{L}^{\prime }$ would form an augmenting $3$-path with respect to $F_{l}$, which would contradict the
maximality of $F_{l}$.
$\square$
\end{proof}

\begin{claim}\label{ChoiceClaim}
\begin{enumerate}
    \item [(1)] for any maximum matchings $F_{L},F_{l}$ of the graph $G$ with $\nu
(G\backslash F_{L})=L(G),$ $\nu (G\backslash F_{l})=l(G)$, we have $(V(F_L)\backslash V(F_l))\cap V_1(G)=\emptyset$;
    \item [(2)] there is a maximum matching $F_{l}$ of $G$ with $\nu (G\backslash F_{l})=l(G)$
    and a maximum matching $F_{L}$ of the graph $G$ with $\nu(G\backslash F_{L})=L(G),$
    such that $V_1(G)\subseteq V(F_L\cap F_l)\cup (V(F_l) \backslash V(F_L))$.
\end{enumerate}
\end{claim}

\begin{proof}(1) On the opposite assumption, consider a vertex $x\in (V(F_L)\backslash V(F_l))\cap V_1(G)$. Since
$x\in V_1(G)$ then $d(x)\leq 2$. On the other hand, (b) of claim \ref{DegreeRequirements} implies that $d(x)\geq 2$,
thus $d(x)=2$. Then there are vertices $y,z$ such that $(x,z)\in F_L$, $(z,y)\in F_l$. Note that due to (a) of
claim \ref{DegreeRequirements}, we have $d(y)\leq 2$. Let us show that $d(y)=1$. Suppose that $d(y)=2$. Then due to (a)
of claim \ref{DegreeRequirements}, we have that $d(z)=2$, thus $G$ is the triangle, which is a contradiction, since $G$
does not satisfy $L(G)=2l(G)$.

Thus $d(y)=1$. Since $x\in V_1(G)$, we imply that there is a vertex $w$ with $d(w)=2$ such that $w,x,z$ form a triangle.
Note that $w$ is covered neither by $F_L$ nor by $F_l$, which contradicts claim \ref{SpanningSubgraph}.

(2) Let $e_t$ be an edge of a triangle $t\in T$ connecting the vertex $v_t\in V_1(G)$ to a vertex of degree two. Let us show that there is a maximum matching $F_{l}$ of $G$ with $\nu (G\backslash F_{l})=l(G)$ such that $e_t\in F_l$ for each $t\in T$.

Choose a maximum matching $F_{l}$ of $G$ with $\nu (G\backslash F_{l})=l(G)$ that contains as many edges $e_t$ as possible. Let us show that $F_l$ contains all edges $e_t$. Suppose that there is $t_0\in T$ such that $e_{t_0}\notin F_l$. Define:
\begin{equation*}
F'_{l}=(F_{l}\backslash \{e\})\cup \{e_{t_0}\},
\end{equation*}%
where $e$ is the edge of $F_l$ that is adjacent to $e_{t_0}$. Note that
\begin{equation*}
\nu(G\backslash F'_l)\leq \nu(G\backslash F_l)=l(G),
\end{equation*}%
thus $F'_{l}$ is a maximum matching of $G$ with $\nu (G\backslash F_{l})=l(G)$. Note that $F'_{l}$ contains more edges $e_t$ than does $F_l$ which contradicts the choice of $F_{l}$.

Thus, there is a maximum matching $F_{l}$ of $G$ with $\nu (G\backslash F_{l})=l(G)$ such that $e_t\in F_l$ for all $t\in T$. Now, for this maximum matching $F_{l}$ of $G$ choose a maximum matching $F_{L}$
of the graph $G$ with $\nu(G\backslash F_{L})=L(G),$ such that $V(F_L\cap F_l)\cup (V(F_l) \backslash V(F_L))$ covers
maximum number of vertices from $V_1(G)$. Let us show that $V_1(G)\subseteq V(F_L\cap F_l)\cup (V(F_l) \backslash V(F_L))$.

Suppose that there is a vertex $x\in V_1(G)$ such that $x\notin V(F_L\cap F_l)\cup (V(F_l) \backslash V(F_L))$. Note that due to claim \ref{SpanningSubgraph} and (b) of claim \ref{DegreeRequirements}, any vertex of degree one is either incident to an edge from $F_L\cap F_l$ or to an edge $V(F_l) \backslash V(F_L)$. Thus due to definition of $V_1(G)$, $d(x)=2$ and if $y$ and $z$ denote the two neighbors of $x$, then $d(y)=2$ and $(y,z)\in E(G)$.

Since $x\notin V(F_L\cap F_l)$, we have that $(x,y)\notin F_L$, and since $x\notin (V(F_l) \backslash V(F_L))$, we have that $(y,z)\notin F_L$, thus $(x,z)\in F_L$, as $F_L$ is a maximum matching. Let $H_L$ be any maximum matching of $G\backslash F_L$. As $L(G)=2l(G)$, we have $(x,y)\in H_L$ ((c1) of theorem \ref{Ratios}). Define: 
\begin{equation*}
F'_{L}=(F_{L}\backslash \{(x,z)\})\cup \{(y,z)\}.
\end{equation*}%
Note that $F'_{L}$ is a maximum matching of $G$, $H_L$ is a matching of $G\backslash F_L$, thus
\begin{equation*}
\nu(G\backslash F'_L)\geq |H_L|=\nu(G\backslash F_L)=L(G).
\end{equation*}%
Therefore $F'_{L}$ is a maximum matching of $G$ with $\nu(G\backslash F'_L)=L(G)$. Now, observe that $V(F'_L\cap F_l)\cup (V(F_l) \backslash V(F'_L))$ covers more vertices than does $V(F_L\cap F_l)\cup (V(F_l) \backslash V(F_L))$ which contradicts the choice of $F_L$. The proof of the claim \ref{ChoiceClaim} is completed. 
$\square$
\end{proof}

\begin{claim}\label{IndependenceClaim} For any maximum matchings $F_{L},F_{l}$ of the graph $G$ with $\nu(G\backslash F_{L})=L(G),$ $\nu (G\backslash F_{l})=l(G)$, we have
\begin{enumerate}
    \item [(1)] $V(F_L)\backslash V(F_l)$ is an independent set;
    \item [(2)] no edge of $G$ connects two vertices that are covered by both $F_L\backslash F_l$ and $F_l\backslash F_L$;
    \item [(3)] no edge of $G$ is adjacent to two different edges from $F_L\cap F_l$;
    \item [(4)] no edge of $G$ connects a vertex covered by $F_L\cap F_l$ to a vertex covered by both $F_L\backslash F_l$ and $F_l\backslash F_L$;
    \item [(5)] if $(u,v)\in F_L\cap F_l$ then either $u\in V_1(G)$ or $v\in V_1(G)$.
\end{enumerate}
\end{claim}

\begin{proof}(1)There is no edge of $G$ connecting two vertices from $%
V(F_{L})\backslash V(F_{l})$ since $F_{l}$ is a maximum matching.

(2) follows from (c1) and (c2) of theorem \ref{Ratios}.

(3) follows from (c3) of theorem \ref{Ratios}.

(4) Suppose that there is an edge $e=(y_{1},y_{2})$, such that $y_1$ is covered by an edge $(z,y_1)\in F_L\cap F_l$ and $y_2$ is covered by both $F_L\backslash F_l$ and $F_l\backslash F_L$. Consider a maximum matching $H_{L}$ of the graph $G\backslash F_{L}.$ Note that $y_{1}$ must be incident to an edge from $H_{L}$, as otherwise we could replace the edge of $H_{L}$ that is adjacent to $e$ and belongs also to $F_{l}\backslash F_{L}$ ((c1) of
theorem \ref{Ratios}) by the edge $e$ to obtain a new maximum
matching $H_{L}^{\prime }$ of the graph $G\backslash F_{L}$ which would not
satisfy (c1) of theorem \ref{Ratios}.

So let $y_{1}$\ be incident to an edge $h_{L}\in H_{L}$, which connects $y_{1}$ with a vertex $x\in V(F_{L})\backslash V(F_{l})$. Note that due to claim \ref{IntersectionEdges}, $z$ is not incident to an edge from $H_L$. Now, let $x_1$ be a vertex such that $(x,x_1)\in F_{L}\backslash F_l$ (such a vertex exists since $x\in V(F_{L})\backslash V(F_{l})$). As $F_l$ is a maximum matching, $x_1$ is incident to an edge $(x_1,x_2)\in F_l\backslash F_L$. By (c1) of theorem \ref{Ratios}, $(x_1,x_2)\in H_L$. Moreover, by claim \ref{AltComp2paths}, $x_2$ is not adjacent to an edge from $F_L$. Thus the edges $(z, y_1), (y_1, x), (x,x_1)$ and $(x_1,x_2)$ form an $F_{L}-H_{L}$ alternating $4$-path $P$. Define:
\begin{eqnarray*}
F_{L}^{\prime } &=&(F_{L}\backslash E(P))\cup (H_{L}\cap E(P)), \\
H_{L}^{\prime } &=&(H_{L}\backslash E(P))\cup (F_{L}\cap E(P)).
\end{eqnarray*}%
Note that $F_{L}^{\prime }$ is a maximum matching of $G$, $H_{L}^{\prime }$
is a matching of $G\backslash F_{L}^{\prime }$ of cardinality $\left\vert
H_{L}\right\vert $, and
\begin{equation*}
\nu (G\backslash F_{L}^{\prime })\geq \left\vert H_{L}^{\prime }\right\vert
=\left\vert H_{L}\right\vert =\nu (G\backslash F_{L})=L(G),
\end{equation*}
thus $H_{L}^{\prime }$ is a maximum matching of $G\backslash F_{L}^{\prime }$
and $\nu (G\backslash F_{L}^{\prime })=L(G)$. This is a contradiction since
the edge $e$ connects two vertices which are covered by $F_{L}^{\prime
}\backslash F_{l}$ and $F_{l}\backslash F_{L}^{\prime }$ ((2) of claim \ref{IndependenceClaim}).

(5)Suppose that $e=(u,v)\in F_{L}\cap F_{l}$. Since $G$ is
connected and $\left\vert V\right\vert \geq 3$, we, without loss of
generality, may assume that $d(v)\geq 2$, and there is $w\in V(G),w\neq u$
such that $(w,v)\in E(G)$. Consider a maximum matching $H_{L}$ of the graph $%
G\backslash F_{L}.$ Note that, without loss of generality, we can assume that $v$ is incident to an edge from $H_{L}$,
as otherwise we could replace the edge of $H_{L}$ that is incident to $w$ ($%
H_{L}$ is a maximum matching of $G\backslash F_{L}$) by the edge $(w,v)$ to
obtain a new maximum matching $H_{L}^{\prime }$ of the graph $G\backslash
F_{L}$ such that $v$\ is incident to an edge from $H_{L}^{\prime }$.

So we can assume that there is an edge $(v,q)\in H_{L}$, $q\neq u$. Note
that due to claim \ref{IntersectionEdges}, $u$ is not incident to an edge from $H_L$. (c3) of theorem \ref{Ratios} implies that $q$ is incident to an edge from $(q,q_1)\in F_L \backslash F_{l}$. As $F_l$ is a maximum matching, $q_1$ is incident to an edge $(q_1,q_2)\in F_l\backslash F_L$. By (c1) of theorem \ref{Ratios}, $(q_1,q_2)\in H_L$. Moreover, by claim \ref{AltComp2paths}, $q_2$ is not adjacent to an edge from $F_L$. Thus the edges $(u, v), (v, q), (q,q_1)$ and $(q_1,q_2)$ form an $F_{L}-H_{L}$ alternating $4$-path $P$. Define:
\begin{eqnarray*}
F_{L}^{\prime } &=&(F_{L}\backslash E(P))\cup (H_{L}\cap E(P)), \\
H_{L}^{\prime } &=&(H_{L}\backslash E(P))\cup (F_{L}\cap E(P)).
\end{eqnarray*}%
Note that $F_{L}^{\prime }$ is a maximum matching of $G$, $H_{L}^{\prime }$
is a matching of $G\backslash F_{L}^{\prime }$ of cardinality $\left\vert
H_{L}\right\vert $, and
\begin{equation*}
\nu (G\backslash F_{L}^{\prime })\geq \left\vert H_{L}^{\prime }\right\vert
=\left\vert H_{L}\right\vert =\nu (G\backslash F_{L})=L(G),
\end{equation*}%
thus $H_{L}^{\prime }$ is a maximum matching of $G\backslash F_{L}^{\prime }$
and $\nu (G\backslash F_{L}^{\prime })=L(G)$. Since $u\in V(F_{l})\backslash
V(F_{L}^{\prime })$ (a) of claim \ref{DegreeRequirements} implies that either $%
d(u)=1$ and therefore $u\in V_1(G)$, or $d(u)=d(v)=2$ and therefore either $u\in V_1(G)$ or $v\in V_1(G)$. Proof of the claim \ref{IndependenceClaim} is completed.
$\square$
\end{proof}

We are ready to complete the proof of the theorem. Take any maximum matchings $F_{L},F_{l}$ of the graph $G$ guaranteed by the (2) of claim \ref{ChoiceClaim} and consider the following partition of $V(G\backslash V_1(G))=V(G)\backslash V_1(G)$:
\begin{eqnarray*}
X =X(F_{L},F_{l})=V(F_{L})\backslash V(F_{l}), Y =Y(F_{L},F_{l})=V(G)\backslash (V_1(G)\cup X).
\end{eqnarray*}

Claim \ref{IndependenceClaim} implies that $X$ and $Y$ are independent sets
of vertices of $G\backslash V_1(G)$, thus $G\backslash V_1(G)$ is a bipartite graph with a bipartition $(X,Y)$.

The choice of maximum matchings $F_{L},F_{l}$, (a) of claim \ref{DegreeRequirements}, (5) of claim \ref{IndependenceClaim} and the definition of the set $Y$ imply (2) of the theorem \ref{Characterization}.

Let us show that it satisfies (3), too.

Consider the alternating $2$-paths of
\begin{equation*}
(H_{L}\backslash F_{l})\triangle (F_{L}\backslash F_{l}).
\end{equation*}%
(c2), (c3) of theorem \ref{Ratios} and the definition of the set $X$
imply that there are $|X|$ such $2$-paths. Moreover, these $2$-paths are in fact $2$-paths of the graph $G\backslash V_{1}(G) $. Thus $G$ satisfies (3) of the theorem. The proof of the theorem 
\ref{Characterization} is completed.
$\square$
\end{proof}

\begin{corollary}
\label{PolynomialAlgorithm} The property of a graph $L(G)=2l(G)$ can be
tested in polynomial time.
\end{corollary}

\begin{proof}
First of all note that the property $L(G)=2l(G)$ is additive, that is, a
graph satisfies this property if and only if all its connected components
do. Thus we can concentrate only on connected graphs.

All connected graphs with $|V(G)|\leq 2$ satisfy the equality $L(G)=2l(G)$, thus we can assume that $|V(G)|\geq 3$.

Next, we construct a set $V_1(G)$, which can be done in linear time. Now, we need to check whether the graph $G\backslash V_1(G)$ satisfies the conditions (1)-(3) of the theorem \ref{Characterization}.

It is well-known that the properties (1) and (2) can be checked in polynomial time, so we will
consider only the testing of (3).

From a graph $G\backslash V_{1}(G)$ with a bipartition $(X,Y)$ we
construct a network $\vec{G}$ with new vertices $s$ and $t$. The arcs of $%
\vec{G}$ are defined as follows:

\begin{itemize}
\item connect $s$ to every vertex of $X$ with an arc of
capacity $2$;

\item connect every vertex of $Y$ to $t$ by an arc of capacity $1$;

\item for every edge $(x,y)\in E(G)$, $x\in X$, $y\in Y$ add
an arc connecting the vertex $x$ to the vertex $y$ which has capacity $1$.
\end{itemize}

Note that

\begin{itemize}
\item the value of the maximum $s-t$ flow in $\vec{G}$ is no more than $%
2\left\vert X\right\vert $ (the capacity of the cut $(S,\bar{%
S})$, where $S=\{s\}$, $\bar{S}=V(\vec{G})\backslash S$, is $2\left\vert
X\right\vert $);

\item the value of the maximum $s-t$ flow in $\vec{G}$ is $2\left\vert
X\right\vert $ if and only if the graph $G\backslash V_{1}(G)$ contains $%
\left\vert X\right\vert $ vertex disjoint $2$-paths,
\end{itemize}

thus (3) also can be tested in polynomial time.
$\square$
\end{proof}

\begin{remark}
Recently Monnot and Toulouse in \cite{Monnot} proved that $2$-path partition
problem remains $NP$-complete even for bipartite graphs of maximum degree
three. Fortunately, in theorem \ref{Characterization} we are dealing with a
special case of this problem which enables us to present a polynomial
algorithm in corollary \ref{PolynomialAlgorithm}.
\end{remark}

\section{$NP$-completeness of testing $L(G)=\frac{3}{2}l(G)$ in the class of bridgeless cubic graphs}

The reader may think that a result analogous to corollary \ref{PolynomialAlgorithm}
can be proved for the property $L(G)=\frac{3}{2}l(G)$ in
the class of graphs containing a perfect matching. Unfortunately this fails already in the class of
bridgeless cubic graphs, which by the well-known theorem of Petersen are known to possess a perfect matching (see theorem 3.4.1 of \cite{Lov}).

\begin{theorem}
It is $NP$-complete to test the property $L(G)=\frac{3}{2}l(G)$ in the class
of bridgeless cubic graphs.
\end{theorem}

\begin{proof}
Clearly, the problem of testing the property $L(G)=\frac{3}{2}l(G)$ for
graphs containing a perfect matching is in $NP$, since if we are given
perfect matchings $F_{L},F_{l}$ of the graph $G$ with $\nu (G\backslash
F_{L})=L(G),$ $\nu (G\backslash F_{l})=l(G)$ then we can calculate $L(G)$
and $l(G)$ in polynomial time.

We will use the well-known $3$-edge-coloring problem (\cite{Holyer}) to
establish the NP-completeness of our problem.

Let $G$ be a bridgeless cubic graph. Consider a bridgeless cubic graph $%
G_{\bigtriangleup }$ obtained from $G$ by replacing every vertex of $G$ by a
triangle. We claim that $G$ is $3$-edge-colorable if and only if $%
L(G_{\bigtriangleup })=\frac{3}{2}l(G_{\bigtriangleup })$.

Suppose that $G$ is $3$-edge-colorable. Then $G_{\bigtriangleup }$
is also $3 $-edge-colorable, which means that $G_{\bigtriangleup }$
contains two edge disjoint perfect matchings $F$ and $F^{\prime }$.
This implies that
\begin{equation*}
L(G_{\bigtriangleup })\geq \nu (G_{\bigtriangleup }\backslash F)\geq
\left\vert F^{\prime }\right\vert =\frac{\left\vert V(G_{\bigtriangleup
})\right\vert }{2},
\end{equation*}%
On the other hand, the set $E(G)$ forms a perfect matching of $%
G_{\bigtriangleup }$, and
\begin{equation*}
l(G_{\bigtriangleup })\leq \nu (G_{\bigtriangleup }\backslash E(G))=\frac{%
\left\vert V(G_{\bigtriangleup })\right\vert }{3},
\end{equation*}%
since every component of $G_{\bigtriangleup }\backslash E(G)$ is a triangle.
Thus:
\begin{equation*}
\frac{L(G_{\bigtriangleup })}{l(G_{\bigtriangleup })}\geq \frac{3}{2},
\end{equation*}%
(d) of theorem \ref{Ratios} implies that $\frac{L(G_{\bigtriangleup })}{%
l(G_{\bigtriangleup })}=\frac{3}{2}$.

Now assume that $\frac{L(G_{\bigtriangleup })}{l(G_{\bigtriangleup })}=\frac{%
3}{2}$. Note that for every perfect matching $F$ of the graph $%
G_{\bigtriangleup }$ the graph $G_{\bigtriangleup }\backslash F$ is a
2-factor, therefore
\begin{eqnarray*}
L(G_{\bigtriangleup }) &=&\frac{\left\vert V(G_{\bigtriangleup })\right\vert
-w(G_{\bigtriangleup })}{2}, \\
l(G_{\bigtriangleup }) &=&\frac{\left\vert V(G_{\bigtriangleup })\right\vert
-W(G_{\bigtriangleup })}{2}
\end{eqnarray*}%
where $w(G_{\bigtriangleup })$ and $W(G_{\bigtriangleup })$ denote the
minimum and maximum number of odd cycles in a $2$-factor of $%
G_{\bigtriangleup }$, respectively. Since $\frac{L(G_{\bigtriangleup })}{%
l(G_{\bigtriangleup })}=\frac{3}{2}$ we have
\begin{equation*}
W(G_{\bigtriangleup })=\frac{\left\vert V(G_{\bigtriangleup })\right\vert
+2w(G_{\bigtriangleup })}{3}.
\end{equation*}%
Taking into account that $W(G_{\bigtriangleup })\leq \frac{\left\vert
V(G_{\bigtriangleup })\right\vert }{3}$, we have:
\begin{eqnarray*}
W(G_{\bigtriangleup }) &=&\frac{\left\vert V(G_{\bigtriangleup })\right\vert
}{3}, \\
w(G_{\bigtriangleup }) &=&0.
\end{eqnarray*}%
Note that $w(G_{\bigtriangleup })=0$ means that $G_{\bigtriangleup }$ is $3$%
-edge-colorable, which in its turn implies that $G$ is $3$-edge-colorable.
The proof of the theorem is completed.
$\square$
\end{proof}

\end{document}